\documentclass[submission,copyright]{eptcs}

\usepackage{amsmath,amsthm,amssymb}
\usepackage{graphicx}
\usepackage{algorithmicx}
  
\usepackage{algpseudocode}
\usepackage{tabto}
\usepackage{moreverb}
\usepackage{multicol}
\usepackage{enumitem}

\usepackage{tikz}
\usetikzlibrary{snakes,arrows,shapes,positioning,calc}

\algrenewcommand{\algorithmiccomment}[1]{\tabto{9cm} // #1}

\newtheorem{theorem}{Theorem}
\newtheorem{definition}[theorem]{Definition}
\newtheorem{lemma}[theorem]{Lemma}

\newtheorem{algorithm}[theorem]{Algorithm}

\newtheorem{observation}[theorem]{Observation}
\newtheorem{corollary}[theorem]{Corollary}

\makeatletter
\newenvironment{proofsketch}[1][\proofname\ sketch]{\par
  \pushQED{\qed}%
  \normalfont \topsep6\p@\@plus6\p@\relax
  \trivlist
  \item[\hskip\labelsep
        \itshape
    #1\@addpunct{.}]\ignorespaces
}{%
  \popQED\endtrivlist\@endpefalse
}
\makeatother

\providecommand{\event}{AFL 2014} 

\title{Analyzing Catastrophic Backtracking Behavior in Practical
  Regular Expression Matching}
\author{Martin Berglund
  \institute{Department of Computing Science,\\
    Ume\aa{} University,\\
    Ume\aa{}, Sweden}
  \email{mbe@cs.umu.se}
  \and
  Frank Drewes
  \institute{Department of Computing Science,\\
    Ume\aa{} University,\\
    Ume\aa{}, Sweden}
  \email{drewes@cs.umu.se}
  \and
  Brink van der Merwe
  \institute{
    Department of Mathematical Sciences,\\
    Computer Science Division,\\
    University of Stellenbosch,\\
    Stellenbosch, South Africa
  }
  \email{abvdm@cs.sun.ac.za}
}

\def\titlerunning{Catastrophic Backtracking in Regular Expression Matching}
\def\authorrunning{M. Berglund, F. Drewes \& B. v.d. Merwe}
\begin{document}
\renewcommand{\L}{\ensuremath{\mathcal{L}}}
\renewcommand{\P}{\ensuremath{\mathcal{P}}}
\newcommand{\sig}{\ensuremath{\Sigma}}
\newcommand{\eps}{\ensuremath{\varepsilon}}
\newcommand{\jlr}[1]{\texttt{P\$#1}}
\long\def\com#1{%
\begin{center}
\fbox{\parbox{.8\linewidth}{%
\textit{}\quad#1
}}
\end{center}
}

\newenvironment{faketheorem}[2]{\par\smallskip\noindent\textbf{#1 \ref{#2}.}\it}{}

\newcommand{\run}{\mathit{run}}
\newcommand{\btr}{\mathit{btr}}
\newcommand{\success}{\mathit{success}}
\newcommand{\acc}{\ensuremath{\textit{Acc}}}
\newcommand{\rej}{\ensuremath{\textit{Rej}}}
\newcommand{\nv}{\ensuremath{\textit{NV}}}
\newcommand{\xto}[1]{\ensuremath{\xrightarrow{\smash{#1}}}}
\newcommand{\td}{\textit{td}}
\newcommand{\sps}{\ensuremath{\flat}}

\newcommand{\nat}{\mathbb N}
\newcommand{\posnat}{\mathbb N_+}
\newcommand{\nxt}{\mathit{next}}
\newcommand{\match}{\textsc{Match}}
\newcommand{\true}{\textit{true}}
\newcommand{\false}{\textit{false}}
\newcommand{\nil}{\textit{nil}}
\newcommand{\Or}{\mathop|}
\newcommand{\stt}{\textit{stt}}
\newcommand{\nfa}[1]{\overline{#1}}
\newcommand{\da}{\mathit{da}}

\long\def\comment#1{}

\newcommand{\fc}[2]{\marginpar{\color{black}\textbf{#2}}\textcolor{red}{#1}}

\newcommand{\Th}{\ensuremath{\textit{Th}}}
\newcommand{\Thp}{\ensuremath{\textit{Th}^p\!}}
\newcommand{\Jp}{J^p\!}

\maketitle

\begin{abstract}
  We develop a formal perspective on how regular expression matching works
  in Java\footnote{Java is a registered trademark of Oracle and/or its
    affiliates. Other names may be trademarks of their respective
    owners.}, a popular representative of the category of
  regex-directed matching engines. In particular, we define an automata model which
  captures all the aspects needed to study such matching engines
  in a formal way.  Based on this, we propose two types of static analysis, which take
  a regular expression and tell whether there exists a family of
  strings which makes Java-style matching run in exponential time.
\end{abstract}

\section{Introduction}
Regular expressions constitute a concise, powerful, and useful pattern
matching language for strings.  They are commonly used to specify
token lexemes for scanner generation during compiler construction, to
validate input for web-based applications, to recognize meaningful
patterns in natural language processing and data mining, for example,
locating e-mail addresses, and to guard against computer system
intrusion.  Libraries for their use are found in most widely-used
programming languages.

There are two fundamentally different types of regex matching engines:
DFA (Deterministic Finite Automaton) and NFA (Non-deterministic Finite
Automaton) matching engines.  DFA matchers are used in (most versions
of) awk, egrep, and in MySQL, and are based on the NFA to DFA subset
conversion algorithm.  This paper deals with NFA engines, which are
found in GNU Emacs, Java, many command line tools, .NET, the PCRE
(Perl compatible regular expressions) library, Perl, PHP, Python, Ruby
and Vim.  NFA matchers make use of an input-directed depth-first
search on an NFA, and thus the matching performed by NFA engines is
referred to as backtracking matching.  NFA engines have made it
possible to extend regular expressions with captures, possessive
quantifiers, and backreferences.

Theory has however not kept pace with practice when it comes to
understanding NFA engines.  We now have NFA matchers that are more
expressive and succinct than the originally developed DFA matchers,
but are also in some cases significantly slower.  Although it is known
that in the worst case, the matching time of NFA matchers is
exponential in the length of input strings~\cite{kirrage:13}, their
performance characteristics and operational matching semantics are
poorly understood in general.  Exponential matching time, also
referred to as catastrophic backtracking (by NFA matchers), can of
course be avoided by using the DFA matchers, but then a less
expressive pattern matching language has to be used. Catastrophic
backtracking has potentially severe security implications, as
denial-of-service attacks are possible in any application which
matches a regular expression to data not carefully controlled by the
application.


This work was motivated by the algorithm presented by
Kirrage et.\ al.\ in \cite{kirrage:13}, which for 
regular expressions with catastrophic backtracking comes up with a
family of strings exhibiting this exponential matching time
behavior. However, they only consider the case where the exponential
matching behavior can be exhibited by strings that are rejected.  We
investigate the complexity of deciding exponential backtracking
matching on strings that are rejected (which we refer to as deciding
exponential failure backtracking) further, and in addition we consider
the general case of exponential backtracking. For this we introduce prioritized
NFA (pNFA), which make non-deterministic choices in an ordered manner, thus
prioritizing some over others in a way very reminiscent of parsing
expression grammars (PEGs). The latter introduce ordered choice to
the world of context-free grammars~\cite{ford:04}. An interesting algorithm
bridging the two areas is given in~\cite{Medeiros:11} by translating extended
regular expressions to PEGs.

By linking failure
backtracking with ambiguity in NFA, we show that catastrophic failure
backtracking can be decided in polynomial time, and in the case of
polynomial failure backtracking, the degree of the polynomial can be
determined in polynomial time. General backtracking is shown decidable
in EXPTIME by associating a tree transducer with the expression and
applying a result from~\cite{drewes:01}.




\section{Preliminaries}\label{prelim}
For a set $A$, we denote by $\P(A)$ the power set of $A$. The constant function $f\colon A\to B$ with $f(a)=b\in B$ for all $a\in A$ is denoted by $b^A$. Also, given any function $f\colon A\to B$ and elements $a\in A$, $b\in B$, we let $f_{a\mapsto b}$ denote the function $f'$ such that $f'(a)=b$ and $f'(x)=f(x)$ for all $x\in A\setminus\{a\}$.
The set of all strings (or sequences) over $A$ is denoted by $A^*$. In particular, it contains the empty string $\eps$. To avoid confusion, it is assumed that $\eps\notin A$. The length of a string $w$ is denoted by $|w|$, and the number of occurrences of $a\in A$ in $w$ is denoted by $|w|_a$. The union of disjoint sets $A$ and $B$ is denoted by $A\uplus B$.

As usual, a regular expression over an alphabet $\sig$ (where
$\varepsilon,\emptyset\notin\sig$) is either an element of
$\sig\cup\{\varepsilon,\emptyset\}$ or an expression of one of the
forms $(E\Or E')$, $(E\cdot E')$, or $(E^*)$, where $E$ and $E'$ are regular expressions. Parentheses can be
dropped using the rule that ${}^*$ (Kleene closure) takes precedence
over $\cdot$ (concatenation), which takes precedence over $\mid$
(union). Moreover, outermost parentheses can be dropped, and $E\cdot
E'$ can be written as $EE'$. The language $\L(E)$ denoted by a regular
expression is obtained by evaluating $E$ as usual, where $\emptyset$
stands for the empty language and $a\in\sig\cup\{\varepsilon\}$ for
$\{a\}$.

A \emph{tree} with labels in a set \sig\ is a function $t\colon V \to
\Sigma$, where $V\subseteq\posnat^*$ is a non-empty, finite set of vertices
(or nodes) which are such that (i) $V$ is prefix-closed, i.e., for all
$v\in\posnat^*$ and $i\in\posnat$, $vi\in V$ implies $v\in V$; and;
(ii) $V$ is closed to the left, i.e., for all $v\in\posnat^*$ and
$i\in\posnat$, $v(i+1)\in V$ implies $vi\in V$.

The vertex $\eps$ is the root of the tree and vertex $vi$ is the $i$th child of $v$.
We let $|t|=|V|$ denote the size of $t$. $t/v$ denotes the tree $t'$
with vertex set $V'=\{w\in\posnat^*\mid vw\in V\}$, where $t'(w) =
t(vw)$ for all $w\in V'$. If $V$ is not explicitly named, we may
denote it by $V(t)$.  The \emph{rank} of a tree $t$ is the maximum number of children of vertices of $t$.
Given trees $t_1,\dots,t_n$ and a symbol $\alpha$, we let $\alpha[t_1,\ldots,t_n]$ denote the
tree $t$ with $t(\eps) = \alpha$ and $t/i = t_i$ for all $i\in
\{1,\ldots,n\}$. The tree $\alpha[]$ may be abbreviated by $\alpha$.

Given an alphabet \sig, the set of all trees of the form $t\colon
V\to\sig$ is denoted by $T_\sig$. Moreover, if $Q$ is an alphabet
disjoint with \sig, we denote by $T_\sig(Q)$ the set of all trees
$t\colon V\to\sig\cup Q$ such that only leaves may be labeled with
symbols in $Q$, i.e., $t(v)\in Q$ implies that $v\cdot1\notin V$.


  A \emph{non-deterministic finite automaton} (NFA) is a tuple
  $A=(Q,\sig,q_0,\delta, F)$ where $Q$ is a finite set of
  states, $\sig$ is an alphabet with $\eps\notin\sig$, $q_0 \in Q$, $F\subseteq Q$ and
  $\delta\colon Q\times (\{\eps \}\cup \sig) \to \P(Q)$ is the
  transition function.
  The fact that $p\in\delta(q,\alpha)$ may also be denoted by
  $q\xrightarrow\alpha p$.
  
  A \emph{run} on a string $w\in\sig^*$ is a sequence
  $p_1\cdots p_{m+1} \in Q^*$ such that there exist
  $\alpha_1,\ldots,\alpha_m \in \sig \cup \{\eps\}$ with
  $\alpha_1\cdots \alpha_m = w$ and $p_i\xrightarrow{\alpha_i} p_{i+1}$ for all $i\in
  \{1,\ldots,m\}$. Such a run is \emph{accepting} if $p_1 = q_0$ and $p_m \in F$.
  The string $w$ is accepted by
  $A$ if and only if there exist an accepting run on $w$.   
  The set of strings in $\sig^*$ that are accepted by $A$ is denoted by $\L(A)$.

  A \emph{string-to-tree transducer} is a tuple $\stt =
  (Q,\sig,\Gamma,q_0,\delta)$, where $\sig$ and $\Gamma$ are the input
  and output alphabets respectively, $Q$ is the finite set of states, $q_0
  \in Q$ is the initial state, and $\delta\colon Q \times \Sigma \to
  T_{\Gamma}(Q)$ is the transition function. When $\delta(q,\alpha) =
  t$ we also write $q \xrightarrow{\alpha} t$.

  For $\alpha_1,\dots,\alpha_n\in \sig$,
  $\stt(\alpha_1\cdots\alpha_n)$ is the set of all trees $t\in
  T_\Gamma$ such that there exists a sequence of trees $t_0, \ldots,
  t_n$ which fulfill the requirement that $t_0 = q_0$ and $t_n=t$; and for
  every $i \in \{1,\ldots,n\}$, $t_i$ is obtained from $t_{i-1}$ by
  replacing every leaf $v$ for which $t_{i-1}(v)\in Q$ with a tree in
  $\delta(t_{i-1}(v),\alpha_i)$, i.e., it holds that
  $t_i/v\in\delta(t_{i-1}(v),\alpha_i)$.

\section{Regular Expression Matching in Java}

Here we will take a look at the algorithm used for matching regular
expressions in Java, using the default \texttt{java.util.regex}
package, and describe in pseudocode how matching is
accomplished in this package. The Java implementation is a good
representative of the class of NFA search matchers. It is both fairly
typical and very consistent across different versions (Java 1.6.0u27 is used to generate figures here).
Many other
implementations behave similarly, e.g.~the popular Perl
Compatible Regular Expressions library (PCRE). We try to capture the essence of the Java matching
procedure as accurately as possible while omitting details, add-ons, and tricks that are irrelevant for the
purpose of this paper.


Let us first describe the Java matcher in some detail.
Readers who are not interested in this description may skip ahead to
the second last paragraph before Algorithm~\ref{algo:java-match}.
The core of the matcher is implemented in
\texttt{java.lang.regex.Pattern}. Given a regular expression, it
constructs an object graph of subclasses of the class
\texttt{java.lang.regex.Pattern\$Node} (we briefly call it
\texttt{Node}, assuming all classes to be inner classes of
\texttt{java.lang.regex.Pattern} unless otherwise stated). \texttt{Node}
objects correspond to states, encapsulating their transitions in
addition, and have one relevant method, \texttt{boolean Node.match(Matcher m, int i, CharSequence s)},
which we will closely
mimic later.
The implicit \texttt{this} pointer corresponds to
the state, \texttt{s} is the entire string, \texttt{i} is the index of the
next symbol to be read. The argument \texttt{m} contains a variety of
book-keeping, notably variables corresponding to $C$ in
Algorithm~\ref{algo:java-match} below, as well as after-the-fact information regarding the
accepting run found. (In contrast, \texttt{match} returns \texttt{true}
if and only if the node can, potentially recursively, match
the remainder of the string).  Every \texttt{Node} contains at least a
pointer \texttt{next} which serves as the ``default'' next transition out
of the node. Let us look at the object graph on the left in
Figure~\ref{fig:java-re-1}.
\begin{figure}[htb]\small
  \vskip-.5em
  \centering
  \begin{tabular}{cc}
    \hskip-3em
    \begin{tikzpicture}
      \node[node distance=0,inner sep=0] (dummy) {};
\begin{scope}[>=latex']
  \node[draw,inner sep=0,right=of dummy] (object1) {\begin{tabular}{l} \textbf{Begin}\end{tabular}};
  \node[draw,inner sep=0,below=of object1] (object2) {\begin{tabular}{l} \textbf{Single}\\c=`a`\end{tabular}};
  \node[draw,inner sep=0,below=of object2,xshift=2.5em] (object3) {\begin{tabular}{l} \textbf{Curly}\\type=0\\cmin=0\\cmax=2147483647\end{tabular}};
  \node[draw,inner sep=0,above=of object3,xshift=2.5em] (object4) {\begin{tabular}{l} \textbf{Single}\\c=`b`\end{tabular}};
  \node[draw,inner sep=0,right=of object3] (object6) {\begin{tabular}{l} \textbf{LastNode}\end{tabular}};
  \node[draw,inner sep=0] (object5) at (object4 -| object6) {\begin{tabular}{l} \textbf{Accept}\\next=null\end{tabular}};
  \draw[->] (object1) edge node[left] {next} (object2);
  \draw[->] (object2) edge node[left] {next} (object2 |- object3.north);
  \draw[<-] (object4) edge node[right] {atom} (object4 |- object3.north);
  \draw[->] (object3) edge node[above] {next} (object6);
  \draw[->] (object4) edge node[above] {next} (object5);
  \draw[->] (object6) edge node[right] {next} (object5);
\end{scope}

    \end{tikzpicture} & \hskip-3em
    \begin{tikzpicture}
      \node[node distance=0,inner sep=0] (dummy) {};
\begin{scope}[>=latex']
  \node[draw,inner sep=0,right=of dummy] (object1) {\begin{tabular}{l} \textbf{Begin}\end{tabular}};
  \node[draw,inner sep=0,right=of object1] (object2) {\begin{tabular}{l} \textbf{Prolog}\end{tabular}};
  \node[draw,inner sep=0,below=of object2] (object3) {\begin{tabular}{l} \textbf{Loop}\\countIndex=1\\beginIndex=0\\cmin=0\\cmax=2147483647\end{tabular}};
  \node[draw,inner sep=0,right=of object3,yshift=-5em,xshift=2em] (object5) {\begin{tabular}{l} \textbf{Branch}\\size=2\end{tabular}};
  \node[draw,inner sep=0,above=of object5] (object6) {\begin{tabular}{l} \textbf{array}\end{tabular}};
  \node[draw,inner sep=0,above=of object6,xshift=-2.5em] (object7) {\begin{tabular}{l} \textbf{Single}\\c=`a`\end{tabular}};
  \node[draw,inner sep=0,above=of object6,xshift=2.5em] (object11) {\begin{tabular}{l} \textbf{Single}\\c=`b`\end{tabular}};
  \node[draw,inner sep=0,below=of object3] (object15) {\begin{tabular}{l} \textbf{LastNode}\end{tabular}};
  \draw[->] (object1) edge node[above] {next} (object2);
  \draw[->] (object2) edge node[left] {loop} (object3);
  \draw[->] (object3) edge[in=180,out=0] node[below left,near end] {body} (object5);
  \draw[->] (object3) edge node[left] {next} (object15);
  \draw[->] (object5) edge node[right] {atoms} (object6);
  \draw[->] (object7) edge[out=180,in=40] node[above,near start] {next}
  ($(object3.north west) !0.8! (object3.north east)$);
  \draw[->] (object11) edge[out=90,looseness=0.6,in=70] node[above]
  {next} ($(object3.north west) !0.6! (object3.north east)$);
  \draw[->] (object6) edge node[left]{element 1} (object7);
  \draw[->] (object6) edge node[right]{element 2} (object11);
\end{scope}

    \end{tikzpicture}
  \end{tabular}
  \caption{The left diagram shows essentially the complete internal
    object graph (i.e.\ internal data-structure) of subclasses of {\tt
      Node} Java constructs for
     $ab^*$. On the right we
    show a simplified version of the corresponding object graph for $(a|b)^*$. In the latter
    all nodes without matching effect (in our limited
    expressions) are removed (e.g.\ the node \textbf{Accept} seen in the more complete
    example on the left).}
  \label{fig:java-re-1}
\end{figure}
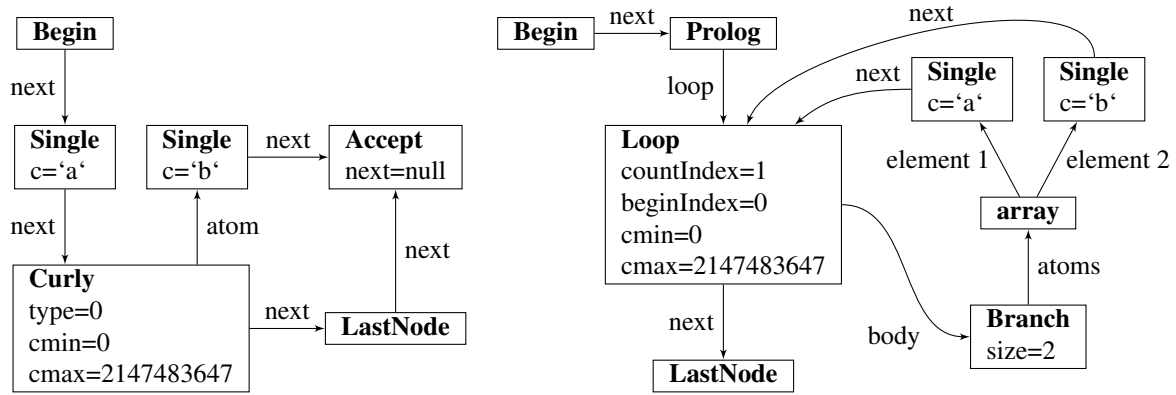
There are quite a few nodes even for a small expression
like $ab^*$, but most are needed for fairly minor book-keeping, and for features
we are not concerned with here. For
example \textbf{LastNode} checks that all symbols are read by the matching, but
can be made to do other things using additional features in {\tt
  java.util.regex} which we do not deal with.

The matching starts with a call to \texttt{match} on \textbf{Begin} with the full
string (i.e., $\texttt{i}$ set to one and the string in $\texttt{s}$). See
Figure~\ref{fig:java-regex-pseudo-code-1} for pseudo-code for the
behavior of \textbf{Begin}, \textbf{Single} and \textbf{Curly}.
\begin{figure}[htb]
  \begin{center}
  \begin{tabular}{c@{\;}c@{\;}c}
    \begin{minipage}{0.3\textwidth}
\begin{algorithmic}[1]
  \small
  \If {$i=0$}
  \State \Return \texttt{next.match}$(i,w)$
  \Else
  \State \Return \false
  \EndIf
\end{algorithmic}
    \end{minipage} & 
    \begin{minipage}{0.34\textwidth}
\begin{algorithmic}[1]
  \small
  \If {$\alpha_i = \texttt{c}$}
  \State \Return \texttt{next.match}$(i+1,w)$
  \Else
  \State \Return \false
  \EndIf
\end{algorithmic}
    \end{minipage} & 
    \begin{minipage}{0.36\textwidth}
\begin{algorithmic}[1]
  \small
  \If {$\texttt{atom.match}(i,w)$}
  \State $b \mathrel{:=} \text{\#symbols read above}$
  \If {$\texttt{this.match}(i+b,w)$}
  \State \Return true
  \EndIf
  \EndIf
  \State \Return $\texttt{next.match}(i,w)$
\end{algorithmic}
\end{minipage}
\end{tabular}
  \end{center}
  \vskip-.5em
  \caption{The code for a call of the form
    $\texttt{match}(i,w=\alpha_1\cdots \alpha_n)$ on a \textbf{Begin} (left),
    \textbf{Single} (middle) and \textbf{Curly} (right) node. \textbf{Single} has a member
    variable \texttt{c} identifying the symbol it should read. \textbf{Curly}
    tries to recursively repeat \texttt{atom}, calling \texttt{next} when
    that fails.}
  \label{fig:java-regex-pseudo-code-1}
\end{figure}
\textbf{Begin} (and \textbf{LastNode}) are trivial, they just check that we are in the
expected position of the string, and in the case of \textbf{Begin} calls
its \texttt{next}. \textbf{Single} reads a single symbol (equal to its internal variable
\texttt{c}) and continues to \texttt{next}. \textbf{Accept} is even more trivial and
always returns true. 
\textbf{Curly}
handles the Kleene closure, and, since it has to resolve
non-determinism (i.e. how many repetitions to perform), it is a bit more
complex. The values \texttt{type}, \texttt{cmin}, and \texttt{cmax} are
irrelevant for our concerns, they implement the counted repetition 
extension.
Note that that line~2 in the right-most code snippet in
Figure~\ref{fig:java-regex-pseudo-code-1} works by updating values in
the ``\texttt{m}'' in-out argument, but we leave that unspecific here.
\textbf{Curly}
starts by trying to match the atom node \texttt{atom} to a prefix of the
string. If it succeeds it calls itself recursively, calling match on
\texttt{this}, to process the remainder. When \texttt{atom} fails to match any
further, \textbf{Curly} instead continues to \texttt{next} (backtracking as
needed). In reality \textbf{Curly} uses imperative loops for efficiency, but it
only serves to achieve a constant speedup and is as such irrelevant
for us. \textbf{Curly} is not used for all Kleene closures, if $b=0$ it would
loop forever, so the construction procedure for the object graph only
uses \textbf{Curly} when it (with a fairly limited decision procedure) can tell
that the contents looped is of constant non-zero length.

Next we look at the more general example on the right side of
Figure~\ref{fig:java-re-1}.  Here there are some additional nodes to
consider. \textbf{Branch} implements the union, and \textbf{Prolog} and \textbf{Loop}
implement the Kleene closure (with \textbf{Prolog} calling {\tt
  matchInit} on \textbf{Loop} to initialize the loop). Let us look at each
function in Figure~\ref{fig:java-regex-pseudo-code-2}.
\begin{figure}
  \vskip-.5em
  \begin{center}
  \begin{tabular}{c@{\;}c@{\;}c}
    \begin{minipage}{0.33\textwidth}
\begin{algorithmic}[1]
  \small
  \State {$\texttt{this.sp} \mathrel{:=} i$}
  \If {\texttt{body.match}$(i,w)$}
  \State \Return true
  \Else
  \State \Return \texttt{next.match}$(i,w)$
  \EndIf
\end{algorithmic}
    \end{minipage} & 
    \begin{minipage}{0.33\textwidth}
\begin{algorithmic}[1]
  \small
  \If {$i > \texttt{this.sp}$}
  \If {\texttt{body.match}$(i,w)$}
  \State \Return true
  \EndIf
  \EndIf
  \State \Return {\texttt{next.match}$(i,w)$}
\end{algorithmic}
    \end{minipage} & 
    \begin{minipage}{0.33\textwidth}
\begin{algorithmic}[1]
  \small
  \For {$e \in \texttt{array}$}
  \If {$e$\texttt{.match}$(i,w)$}
  \State \Return true
  \EndIf
  \EndFor
  \State \Return \false
\end{algorithmic}
\end{minipage}
\end{tabular}
  \end{center}
  \vskip-.25em
  \caption{The code for a call of the form
    $\texttt{match}(i,w=\alpha_1\cdots \alpha_n)$ on a \textbf{Loop}
    (middle), \textbf{Branch} (right) and, as a special case, the call
    $\texttt{matchInit}(i,w)$ on \textbf{Loop} on the
    left. \texttt{matchInit} is called by \textbf{Prolog} in lieu of
    \texttt{match}. Notice that the loop in \textbf{Branch} is
    \emph{in array order}.}
  \label{fig:java-regex-pseudo-code-2}
\end{figure}
In \textbf{Branch}, \texttt{match} starts by letting the first
subexpression match, continuing with the second and so on if the first
attempts fail. The symbiotic relationship between \textbf{Prolog} and \textbf{Loop} is
trickier. Where all other nodes calls into \textbf{Loop} with ${\tt
  match}(i,w)$ as usual \textbf{Prolog} calls in with $\texttt{matchInit}(i,w)$
(on the left in Figure~\ref{fig:java-regex-pseudo-code-2}). This
serves only one purpose: it eliminates $\eps$-cycles. That is, it
prevents \textbf{Loop} from recursively matching $\texttt{body}$ to the empty
string, making no progress. In $\texttt{matchInit}$ the current value of
$i$ is stored, and in $\texttt{match}$ (in the middle in
Figure~\ref{fig:java-regex-pseudo-code-2}) an attempt to match $\texttt{body}$ will
only be made if at least one symbol has been read since
the last attempt.

As an additional example, consider the regular expression $(a|a)^*$,
which has an object graph almost like on the right of
Figure~\ref{fig:java-re-1}, except the second \textbf{Single} \emph{also} has
\texttt{c} set to $a$. Matching this against $aa\cdots ab$ will take
exponential time in the number of $a$s, as all ways to match each $a$
to each \textbf{Single} in $a|a$ will be tried as the matching backtracks
trying to match the final $b$. In an experiment on one of the authors'
desktop PCs an attempt to match $(a|a)^*$ to $a^{35}b$ using Java took
roughly an hour of CPU time.

The object graph on the right in Figure~\ref{fig:java-re-1} is, as is
noted in the caption, a bit of creative editing of reality. A number
of nodes not affecting the search behavior or matching are removed:
the \textbf{Accept} node, which is just a \texttt{next} placeholder with no
effect, \textbf{GroupHead} and \textbf{GroupTail}, which tracks what part of the match
corresponds to a parenthesized subexpression, and finally \textbf{BranchConn},
which is placed in relation to \textbf{Branch} in the right of
Figure~\ref{fig:java-re-1} and records some information for the
optimizer.

In general all nodes have numerous additional features not discussed,
and there are many additional nodes serving similar purposes. For
example \textbf{Single} may be replaced with \textbf{Slice} to match multiple symbols at
once or \textbf{BnM} to matche multiple symbols using Boyer-Moore
matching~\cite{boyer1977fast}. However, the optimizations are too
minor to matter for our concerns (e.g.~using \textbf{Slice}
and \textbf{BnM} instead of \textbf{Single} yields at most a linear speed-up), and the
additional features are outside our scope.


Let us take the above together and assemble the snippets
of code into a function which takes a regular expression and
a string as input and decides if the expression matches the string. A
regular expression is represented by its parse tree, $T\colon
\posnat^* \to \{|, {}^{*}, {}^{*?}, \cdot,\epsilon\} \cup \Sigma$, defined in the
obvious way with each $\cdot$ and $|$ having two children,
${}^*$ and ${}^{*?}$ one, and $\alpha \in \sig\cup\{\epsilon\}$ zero.
The operator ${}^{*?}$ is the lazy Kleene closure, which is the same as ${}^*$,
except that it attempts to make as \emph{few} repetitions as possible.

We now define a function $\nxt : \posnat^* \to \posnat^*\uplus\{\nil\}$ on the nodes of $T$, where $\nil$ is a special value.
Roughly speaking, $\nxt(v)$ is the node at which parsing continues when the subexpression
rooted av $v$ has successfully matched (compare to the \texttt{cont} pointers in Kirrage at
al.~\cite{kirrage:13}). Let $\nxt(\eps)=\nil$, and 
\begin{enumerate}
\item if $T(v)=|$ then $\nxt(v\cdot 1)=\nxt(v\cdot 2)=\nxt(v)$;
\item if $T(v)=\cdot$ then $\nxt(v\cdot 1)=v\cdot 2$ and $\nxt(v\cdot 2)=\nxt(v)$; and
\item if $T(v)={}^{*}$ or $T(v) = {}^{*?}$ then $\nxt(v\cdot 1)=v$.
\end{enumerate}
Then, collapsing the object graph and ignoring precise node choices in
Java we get Algorithm~\ref{algo:java-match}.
\begin{algorithm}
  \label{algo:java-match}
  Simplified pseudocode of the Java matching algorithm. The
  implicit regular expression parse tree is $T$. The
  call-by-value input parameters are the node of $T$ currently
  processed, the remainder of the input string, and a set of
  nodes that we should not revisit before consuming the next input symbol.
  This prevents $\eps$-cycles as discussed above. The initial call made is
  $\match(\eps, w, \emptyset)$.
  \begin{multicols}{2}
  \begin{algorithmic}[1]
    \small
  \Function{\match}{$v,w=a_1\cdots a_n,C$} 
    \If {$v=\nil$}
      \State {\Return $n=0$} 
    \ElsIf {$T(v)=\epsilon$}
      \State {\Return \textsc{Match}($\nxt(v),w,C$)}
    \ElsIf {$T(v)\in\Sigma$}
      \If {$n\ge 1 \wedge T(v)=a_1$}
        \State{\Return \textsc{Match}($\nxt(v),a_2\cdots
          a_n,\emptyset$)}
        \EndIf
        \State{\Return \false}
    \ElsIf {$T(v)=|$}
      \If {$\match(v\cdot 1,w,C)$}
        \State{\Return \true}
      \EndIf
        \State{\Return $\match(v\cdot 2,w,C)$}
    \ElsIf {$T(v)=\cdot$}
      \State {\Return $\match(v\cdot 1,w,C)$}
    \ElsIf {$T(v)={}^*$}
      \If {$v\cdot 1 \notin C$}
        \If {$\match(v\cdot 1,w,C\cup\{v \cdot 1\})$}
          \State {\Return \true}
        \EndIf
      \EndIf
        \State {\Return $\match(\nxt(v),w,C)$}
    \ElsIf {$T(v)={}^{*?}$}
      \If {$\match(\nxt(v),w,C)$}
        \State {\Return \true}
      \ElsIf {$v\cdot 1 \notin C$}
        \State {\Return $\match(v\cdot 1,w,C\cup\{v\cdot 1\})$}
      \Else
        \State {\Return \false}
      \EndIf
    \EndIf
    \EndFunction
  \end{algorithmic}
  \end{multicols}

\end{algorithm}
Notice how the code for the two Kleene closure variants \emph{only}
differs in what they try first: $\mathord{{}^*}$ tries to repeat its
body first, whereas $\mathord{{}^{*?}}$ tries to not repeat the body.
Note also how $C$ is used to prevent $\eps$-cycles in
lines 19--20 and 28--29. If the node we \emph{would}
go to is already in $C$ this means that no symbol has been read since
last time we tried this, meaning repeating it would be a loop without
progress.

\section{Prioritized Non-Deterministic Finite Automata}
We now define a modified type of NFA that provides us
with an abstract view of the matching procedure discussed in the
previous section. The modifications have no impact on the
language accepted, but make the automaton ``run deterministic''.
Every string in the language accepted has a unique
accepting run, a property brought about by ordering the
non-deterministic choices into a first, second, etc
alternative, and letting the unique accepting run be given by trying,
at any given state, alternative $i+1$ only when alternative $i$ has failed.
In our definition, only $\eps$-transitions can be nondeterministic.

\begin{definition} 
  \label{defn:pNFA}%
  A \emph{prioritized non-deterministic finite automaton} (\emph{pNFA}) is a tuple $A=(Q_1, Q_2, \sig, q_0,
  \delta_1, \delta_2,\allowbreak F)$, where $Q_1$ and $Q_2$ are disjoint finite sets of states; $\sig$ is a finite alphabet; $q_0 \in Q_1 \cup Q_2$
  is the initial state; $\delta_1\colon Q_1 \times \sig \to (Q_1 \cup
  Q_2)$ is the deterministic transition function; $\delta_2\colon Q_2
  \to (Q_1 \cup Q_2)^*$ is the non-deterministic prioritized
  transition function; and $F \subseteq Q_1 \cup Q_2$ are the final
  states.
  
  The NFA corresponding to the pNFA $A$ is given by $\nfa A = (Q_1\cup Q_2,\sig, q_0, \bar{\delta},F)$, where
  \[
  \bar{\delta}(q,\alpha) = \left\{\begin{array}{cl}
      \{\delta_1(q,\alpha)\} & \text{if $q\in Q_1$ and $\alpha\in
        \sig$,}\\
      \{q_1,\dots,q_n\} & \text{if $q\in Q_2$, $\alpha=\eps$, and
        $\delta_2(q) = q_1\cdots q_n$.}
      \end{array}\right.
    \]
    The \emph{language accepted by $A$}, denoted by $\L(A)$, is $\L(\nfa A)$.
\end{definition}

Next, we define the so-called backtracking run of a pNFA on an input string $w$. This run takes the form of a tree which, intuitively, represents the attempts a matching algorithm such as Algorithm~\ref{algo:java-match} would make until accepting the input string (or eventually rejecting it). The definition makes use of a parameter $C$ whose purpose is to remember, for every state, the highest nondeterministic alternative that has been tried since the last symbol was consumed. This corresponds to the parameter $C$ in Algorithm~\ref{algo:java-match} and avoids infinite runs caused by $\eps$-cycles.

\begin{definition}  
  \label{defn:btr}
  Let $A=(Q_1,Q_2,\sig,q_0,\delta_1,\delta_2,F)$ be a pNFA, $q\in
  Q_1\cup Q_2$, $w=\alpha_1\cdots
  \alpha_n\in\sig^*$, and $C\colon Q_2\to\nat$. Then the \emph{$(q,w,C)$-backtracking run} of $A$ is
  a tree over $Q_1\cup Q_2 \uplus \{\acc,\rej\}$. It \emph{succeeds} if and only if
  $\acc$ occurs in it.    We denote the $(q,w,C)$-backtracking run by
  $\btr_A(q,w,C)$ and inductively define it as follows. If $q\in F$ and $w=\eps$ then $\btr_A(q,w,C) = q[\acc]$. Otherwise, we distinguish between two cases:\footnote{For the first case, recall that $0^{Q_2}$ denotes the function $C\colon Q_2\to\nat$ such that $C(q)=0$ for all $q\in Q_2$.}
\begin{enumerate}
\item If $q\in Q_1$, then
  \[
  \btr_A(q,w,C)=\left\{\begin{array}{lll}
      q[\btr_A(\delta_1(q,\alpha_1),\alpha_2\cdots \alpha_n,0^{Q_2})]
      & \text{if $n>0$ and $\delta_1(q,\alpha_1)$ is defined,}\\
      q[\rej] & \text{otherwise.}
    \end{array}\right.
  \]
\item If $q\in Q_2$ with $\delta_2(q)=q_1\cdots q_k$, let $i_0=C(q)+1$ and
  $r_i=\btr_A(q_i,w,C_{q\mapsto i})$ for $i_0\le i \le k$. Then
  \[\btr_A(q,w,C)=\left\{\begin{array}{ll}
      q[\rej] & \text{if $i_0>k$,}\\
      q[r_{i_0},\ldots, r_k] & \text{if $i_0\le k$ but no $r_i$ ($i_0\le i \le k$) succeeds,}\\
      q[r_{i_0},\ldots, r_i] & \text{if $i\in\{i_0,\dots,k\}$ is the least index such that
        $r_i$ succeeds.}
\end{array}\right.
\]
\end{enumerate}
The \emph{backtracking run of $A$ on $w$} is
$\btr_A(w)=\btr_A(q_0,w,0^{Q_2})$. If $\btr_A(w)$ succeeds, then the
\emph{accepting run of $A$ on $w$} is the sequence of states on the
right-most path in $\btr_A(w)$.
\end{definition}

Notice that the third parameter $C$ in $\btr_A(q,w,C)$ fulfills a similar purpose as the set $C$ in Algorithm~\ref{algo:java-match}. It is used to track transitions that must not be revisited to avoid cycles.

Clearly, for a pNFA $A$ and a string $w$, $w\in\L(A)$ if and only if $\btr_A(w)$ succeeds, if and only if the accepting run of $A$ on $w$ is an accepting run of the NFA $\nfa A$. Backtracking runs capture the behavior of the following algorithm which generalizes Algorithm~\ref{algo:java-match} to arbitrary pNFAs to deterministically find the accepting run of $A$ on $w$ if it exists.

%

\begin{algorithm}  
  \label{algo:match}%
  Let $A=(Q_1,Q_2,\sig,q_0,\delta_1,\delta_2,F)$ be a pNFA. The call
  $\textsc{Match}(q_0,w,0^{Q_2})$ of the following procedure yields the accepting
  run of $A$ on $w$ if it exists, and $\bot \notin Q_1 \cup Q_2$ otherwise. The third parameter
  is similar to the $C$ in Definition~\ref{defn:btr}. For every state $q\in Q_2$ with
  out-degree $d$ we have $C(q)\in\{0,\dots,d\}$.
  \begin{multicols}{2}
  \begin{algorithmic}[1]
    \small
    \Function{Match}{$q$,$w=a_1\cdots a_n$,$C$}
    \If { $q\in Q_1$ }
    \If { $n=0$ }
    \If {$q\in F$}
    \State \Return $q$
    \Else
    \State \Return $\bot$ 
    \EndIf
    \Else
    \State \Return $q \cdot \textsc{Match}(\delta_1(q,a_1),a_2\cdots
    a_n,0^{Q_2})$ 
    \EndIf
    \Else 
    \If { $n=0\wedge q\in F$} 
    \State \Return q  
    \Else
    \State $q_1\cdots q_k \mathrel{:=} \delta_2(q)$
    \For {$i=C(q)+1,\dots,k$}
      \State $r \mathrel{:=}\textsc{Match}(q_i,w,C_{q\mapsto i})$
      \If {$r\neq\bot$}
        \State \Return $q\cdot r$ 
      \EndIf
    \EndFor
    \State \Return $\bot$ 
    \EndIf
    \EndIf
    \EndFunction
  \end{algorithmic}
  \end{multicols}
  \noindent 
  Notice especially line 10 where a symbol is read and $C$ is reset to
  $0^{Q_2}$ in the recursive call. The case for $Q_2$ starts at line
  13, the loop at 17 tries all \emph{not yet tried} transitions for
  that state. If no transition succeeds we fail on line 23.
\end{algorithm}

We note here that the running time of Algorithm~\ref{algo:match} is exponential in general, just like Algorithm~\ref{algo:java-match}. This can be remedied by means of memoization, but potentially with a significant memory overhead, due to the fact that memoization needs to keep track of each possible assignment to all $C(q)$ with $q\in Q_2$.\footnote{Apparently, starting from version~5.6, Perl uses memoization in its regular expression engine in order to speed up matching.}

Depending on how one turns a given regular expression into a pNFA, Algorithm~\ref{algo:match} will run more or less efficiently. For example, if the pNFA is built in a way that reflects Algorithm~\ref{algo:java-match}, analyzing the efficiency of Algorithm~\ref{algo:match} or, equivalently, the size of backtracking runs, yields a (somewhat idealized) statement about the efficiency of the Java matcher.
  
\subsection{Two Constructions for Turning Regular Expressions into pNFA}

In this section we give two examples of constructions that can be used to turn a regular expression $E$ into a pNFA $A$ such that $\L(A)=\L(E)$. The first is a prioritized version of the classical Thompson construction \cite{Thompson:68}, whereas the second follows the Java approach. 

Recall that the classical Thompson construction converts the parse tree $T$ of a regular expression $E$ to an NFA, which we denote by $\Th(E)$, by doing a postorder traversal on $T$. An NFA is constructed for each subtree $T'$ of $T$, equivalent to the regular expression represented by $T'$. We do not repeat this well-known construction here, assuming that the reader is familiar with it. Instead, we define a prioritized version, which constructs a pNFA denoted by $\Thp(E)$ such that $\nfa{\Thp(E)}=\Th(E)$.

Just as the construction for $\Th(E)$, we define $\Thp(E)$ recursively on the parse tree for $E$. For each subexpression $F$ of $E$, 
$\Thp(F)$ has a single initial state with no ingoing transitions, and a single final state with no outgoing transitions.
The constructions of $\Thp(\emptyset)$, $\Thp(\eps)$, $\Thp(a)$, and $\Thp(F_1\cdot F_2)$, given that $\Thp(F_1)$ and $\Thp(F_2)$ are already constructed, are defined as for $\Th(E)$, splitting the state set into $Q_1$ and $Q_2$ in the obvious way. It is only when we construct $\Thp(F_1 | F_2)$ from $\Thp(F_1)$ and $\Thp(F_2)$, and $\Thp(F_1^*)$ from $\Th(F_1)$, where the priorities of introduced $\eps$-transitions require attention. We also consider the lazy Kleene closure $F_1^{*?}$, to illustrate the difference in priorities of transitions between constructions for the greedy and lazy Kleene closure.
In each of the constructions below, we assume that $\Thp(F_i)$ ($i\in\{1,2\}$) has the initial state $q_i$ and the final state $f_i$. Furthermore, $\delta_2$ denotes the transition function for $\eps$-transitions in the 
newly constructed pNFA $\Thp(E)$. All non-final states in $\Thp(E)$ that are in $\Thp(F_i)$ inherit their outgoing transitions from $\Thp(F_i)$.

\begin{itemize}
\item If $E = F_1|F_2$ then $\Thp(E)$ is built like $\Th(E)$, thus introducing new initial and final states $q_0$ and $f_0$, respectively, and defining $\delta_2(q_0)=q_1q_2$ and $\delta_2(f_1)=\delta_2(f_2)=f_0$.
\item If $E = F_1^*$ then we add new initial and final states $q_0$
  and $f_0$ to $Q_2$ and define  $\delta_2(q_0)=q_1f_0$ and
  $\delta_2(f_1)=q_1f_0$. The case $E=F_1^{*?}$ is the same, except
  that $\delta_2(q_0)=f_0q_1$ and $\delta_2(f_1)=f_0q_1$.
\end{itemize}

Thus, the pNFA $\Thp(F^*)$ tries $F$ as often as possible whereas $\Thp(F^{*?})$ does the opposite.

The second pNFA construction is the one implicit in the Java approach and Algorithm~\ref{algo:java-match}. We denote this pNFA by $\Jp(E)$. The base cases $\Jp(\emptyset)$, $\Jp(\eps)$, $\Jp(a)$ are identical to $\Thp(\emptyset)$, $\Thp(\eps)$, $\Thp(a)$, respectively. Now, let us consider the remaining operators. Again, we assume that $\Jp(F_i)$ ($i\in\{1,2\}$) has the initial state $q_i$ and the final state $f_i$. Furthermore, $\delta_2$ denotes the transition function for $\eps$-transitions in the 
newly constructed pNFA $\Jp(E)$.

\begin{itemize}
\item Assume that $E = F_1\cdot F_2$. Then $\Jp(E)$ is built from $\Jp(F_1)$ and $\Jp(F_2)$ by identifying $f_1$ with $q_2$, adding a new initial state $q_0\in Q_2$ with $\delta_2(q_0)=q_1$, and making $f_2$ the final state. Thus, $\Jp(E)$ is built like $\Thp(E)$, except that a new initial state is added and connected to the initial state of $\Jp(F_1)$ by means of an $\eps$-transition.
\item If $E = F_1|F_2$ then $\Jp(E)$ is constructed by introducing a new initial state $q_0$, defining $\delta_2(q_0)=q_1q_2$, and identifying $f_1$ and $f_2$, the result of which becomes the new final state.
\item Now assume that $E = F_1^*$. Then we add a new final state $f_0$
  to $\Jp(F_1)$, make $q_0=f_1$ the initial state of $J^P(E)$, and set
  $\delta_2(q_0)=q_1f_0$. The case $E = F_1^{*?}$ is exactly the same, except that $\delta_2(q_0)=f_0q_1$.
\end{itemize}

\begin{observation}
\label{obs:btr-eqv-runt}
Let $E$ be a regular expression and $A$ a pNFA. Then the running time of Algorithm~\ref{algo:match} on $w$ (with respect to $E$) is $\Theta(|\btr_A(w)|)$.
\end{observation}

The two variants of implementing regular expressions by pNFA are closely related. In fact, Kirrage et al.~\cite{kirrage:13} seem to regard them as being essentially identical and write that their reasons for choosing $\nfa{\Jp(E)}$ are ``purely of presentational nature''. However, using our notion of pNFA we can show that this is not always the case. For this, note first that the construction of both $\Thp(E)$ and $\Jp(E)$ can be viewed in a top-down fashion, where each operation is represented by an abstract pNFA in which zero, one, or two transitions are labeled with regular expressions. Replacing such a transition with the corresponding pNFA yields the constructed pNFA for the whole expression. Figure~\ref{fig:building blocks} shows the building blocks for the operations $\cdot$, $|$, ${}^*$, and ${}^{*?}$ in both cases. Priorities follow the convention that $\eps$-transitions leaving a state are drawn in clockwise order, starting at noon. Unlabeled edges denote $\eps$-transitions.%
\begin{figure}
\vskip-0.5em
\centering
\begin{tabular}{cccc}
\begin{tikzpicture}
  \node (i) at (0,0) {};
  \draw[white] (0,-1cm) -- (0,1cm);
  \begin{scope}[node distance=1.8em,every node/.style={draw,circle,inner sep=2.5pt}]
    \node[right=of i,xshift=-1em] (lp) {};
    \node[right=of lp] (s) {};
    \node[right=of s] (sp) {};
    \draw (sp) circle (2pt);
  \end{scope}
  \begin{scope}[->,>=latex']
    \draw (i) edge (lp);
    \draw (lp) edge node[above] {$E_1$} (s);
    \draw (s) edge node[above] {$E_2$} (sp);
  \end{scope}
\end{tikzpicture}
&
\begin{tikzpicture}
  \node (i) at (0,0) {};
  \draw[white] (0,-1cm) -- (0,1cm);
  \tikzset{close/.style={node distance=1.8em}}
  \begin{scope}[node distance=1.8em,every node/.style={draw,circle,inner sep=2.5pt}]
    \node[right=of i,xshift=-1em] (lp) {};
    \node[close,above right=of lp] (s) {};
    \node[close,below right=of lp] (s2) {};
    \node[right=of s] (p) {};
    \node[right=of s2] (p2) {};
    \node[close,below right=of p] (e) {};
    \draw (e) circle (2pt);
  \end{scope}
  \begin{scope}[->,>=latex']
    \draw (i) edge (lp);
    \draw (lp) edge (s);
    \draw (lp) edge (s2);
    \draw (s) edge node[above] {$E_1$} (p);
    \draw (s2) edge node[above] {$E_2$} (p2);
    \draw (p) edge (e);
    \draw (p2) edge (e);
  \end{scope}
\end{tikzpicture} &
\begin{tikzpicture}
  \tikzset{close/.style={node distance=1.8em}}
  \node (i) at (0,0) {};
  \draw[white] (0,-1cm) -- (0,1cm);
  \begin{scope}[node distance=1.8em,every node/.style={draw,circle,inner sep=2.5pt}]
    \node[right=of i,xshift=-1em] (n1) {};
    \node[close,right=of n1] (n2) {};
    \node[right=of n2] (n3) {};
    \node[close,right=of n3] (n4) {};
    \draw (n4) circle (2pt);
  \end{scope}
  \begin{scope}[->,>=latex']
    \draw (i) edge (n1);
    \draw (n1) edge (n2);
    \draw (n2) edge node[below] {$E_1$} (n3);
    \draw (n3) edge (n4);
    \draw (n1) edge[bend right=50] (n4);
    \draw (n3) edge[out=80,in=30] (n2);
  \end{scope}
\end{tikzpicture} &
\begin{tikzpicture}
  \node (i) {};
  \begin{scope}[node distance=1.8em,every node/.style={draw,circle,inner sep=2.5pt}]
    \node[right=of i,xshift=-1em] (n1) {};
    \node[right=of n1] (n2) {};
    \node[right=of n2] (n3) {};
    \node[right=of n3] (n4) {};
    \draw (n4) circle (2pt);
  \end{scope}
  \begin{scope}[->,>=latex']
    \draw (i) edge (n1);
    \draw (n1) edge (n2);
    \draw (n2) edge node[above] {$E_1$} (n3);
    \draw (n3) edge (n4);
    \draw (n1) edge[bend left=50] (n4);
    \draw (n3) edge[bend left=50] (n2);
  \end{scope}
\end{tikzpicture}\\
\begin{tikzpicture}
  \node (i) at (0,0) {};
  \draw[white] (0,-1cm) -- (0,1cm);
  \begin{scope}[node distance=1.8em,every
      node/.style={draw,circle,inner sep=2.5pt}]
    \node[right=of i,xshift=-1em] (i2) {};
    \node[right=of i2] (lp) {};
    \node[right=of lp] (s) {};
    \node[right=of s] (sp) {};
    \draw (sp) circle (2pt);
  \end{scope}
  \begin{scope}[->,>=latex']
    \draw (i) edge (i2);
    \draw (i2) edge (lp);
    \draw (lp) edge node[above] {$E_1$} (s);
    \draw (s) edge node[above] {$E_2$} (sp);
  \end{scope}
\end{tikzpicture}
&
\begin{tikzpicture}
  \tikzset{close/.style={node distance=1.8em}}
  \node (i) at (0,0) {};
  \draw[white] (0,-1cm) -- (0,1cm);
  \begin{scope}[node distance=1.8em,every node/.style={draw,circle,inner sep=2.5pt}]
    \node[right=of i,xshift=-1em] (lp) {};
    \node[close,above right=of lp] (s) {};
    \node[close,below right=of lp] (s2) {};
    \node[close,below right=of p] (e) {};
    \draw (e) circle (2pt);
  \end{scope}
  \begin{scope}[->,>=latex']
    \draw (i) edge (lp);
    \draw (lp) edge (s);
    \draw (lp) edge (s2);
    \draw (s) edge[bend left] node[above] {$E_1$} (e);
    \draw (s2) edge[bend right] node[above] {$E_2$} (e);
  \end{scope}
\end{tikzpicture} &
\begin{tikzpicture}
  \tikzset{close/.style={node distance=1.8em}}
  \node (i) at (0,0) {};
  \draw[white] (0,-1cm) -- (0,1cm);
  \begin{scope}[node distance=1.8em,every node/.style={draw,circle,inner sep=2.5pt}]
    \node[right=of i,xshift=-1em] (n1) {};
    \node[right=of n1,xshift=1em] (n2) {};
    \node[below=of n2,yshift=1em] (e) {};
    \draw (e) circle (2pt);
  \end{scope}
  \begin{scope}[->,>=latex']
    \draw (i) edge (n1);
    \draw (n1) edge (n2);
    \draw (n2) edge[bend right=70] node[above] {$E_1$} (n1);
    \draw (n1) edge[bend right] (e);
  \end{scope}
\end{tikzpicture} &
\begin{tikzpicture}
  \tikzset{close/.style={node distance=1.8em}}
  \node (i) at (0,0) {};
  \draw[white] (0,-1cm) -- (0,1cm);
  \begin{scope}[node distance=1.8em,every node/.style={draw,circle,inner sep=2.5pt}]
    \node[right=of i,xshift=-1em] (n1) {};
    \node[right=of n1,xshift=1em] (n2) {};
    \node[above=of n2,yshift=-1em] (e) {};
    \draw (e) circle (2pt);
  \end{scope}
  \begin{scope}[->,>=latex']
    \draw (i) edge (n1);
    \draw (n1) edge (n2);
    \draw (n2) edge[bend left=70] node[below] {$E_1$} (n1);
    \draw (n1) edge[bend left] (e);
  \end{scope}
\end{tikzpicture}
\end{tabular}
\vskip-0.5em
\caption{Abstract pNFA corresponding to $E_1\cdot E_2$, $E_1\Or E_2$,
  $E_1^*$ and $E_1^{*?}$, from which $\Thp(E)$ (top row) and $\Jp(E)$
  (bottom row) are constructed. The transitions are prioritized in
  clockwise order, starting at noon.\label{fig:building blocks}}
\end{figure}
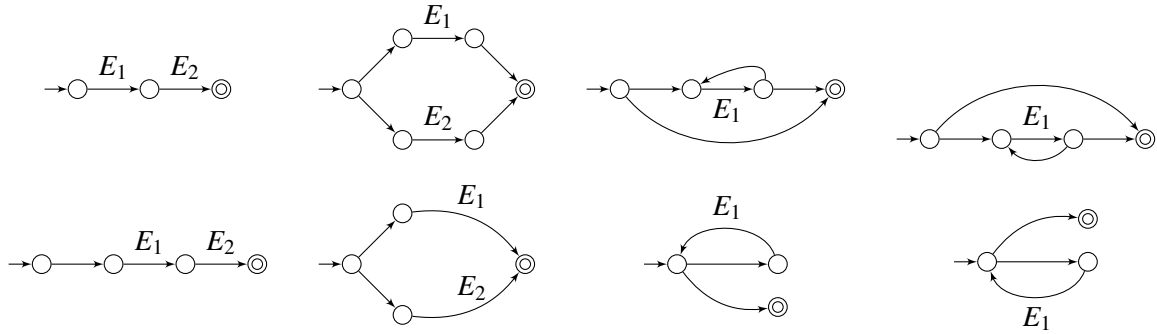

Now consider an expression $E$ of the form
$((\eps|E_1)\cdot\eps^*)^*\cdot E_2$. When building $\Thp(E)$ and
$\Jp(E)$, these correspond to the following abstract pNFA:

\vskip-0.5em
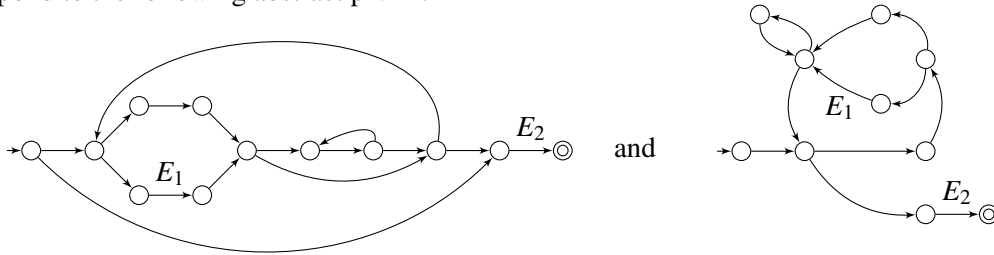
\begin{tabular}{ccc}
\begin{tikzpicture}
  \node (i) at (0,0) {};
  \draw[white] (1cm,-1.75cm) -- (1cm,1.5cm);
  \begin{scope}[node distance=1.5em,every node/.style={draw,circle,inner sep=2.5pt}]
    \node[right=of i,xshift=-1em] (n1) {};
    \node[right=of n1] (n2) {};
    \node[above right=of n2] (n3) {};
    \node[below right=of n2] (n4) {};
    \node[right=of n3] (n5) {};
    \node[right=of n4] (n6) {};
    \node[below right=of n5] (n7) {};
    \node[right=of n7] (n8) {};
    \node[right=of n8] (n9) {};
    \node[right=of n9] (n10) {};
    \node[right=of n10] (n11) {};
    \node[right=of n11] (n12) {};
    \draw (n12) circle (2pt);
  \end{scope}
  \begin{scope}[->,>=latex']
    \draw (i) edge (n1);
    \draw (n1) edge (n2);
    \draw (n2) edge (n3);
    \draw (n2) edge (n4);
    \draw (n3) edge (n5);
    \draw (n4) edge node[above] {$E_1$} (n6);
    \draw (n5) edge (n7);
    \draw (n6) edge (n7);
    \draw (n7) edge (n8);
    \draw (n8) edge (n9);
    \draw (n9) edge (n10);
    \draw (n10) edge (n11);
    \draw (n11) edge node[above] {$E_2$} (n12);
    \draw (n1) edge[bend right=45] (n11);
    \draw (n7) edge[bend right] (n10);
    \draw (n9) edge[out=80,in=30] (n8);
    \draw (n10) edge[out=80,in=80] (n2);
  \end{scope}
\end{tikzpicture} & 
\begin{tikzpicture}
  \draw[white] (0,-1.8cm) -- (0,2cm);
  \node at (0,0) {and};
\end{tikzpicture}&
\begin{tikzpicture}
  \node (i) at (0,0) {};
  \draw[white] (0,-1.75cm) -- (0,1.5cm);
  \begin{scope}[node distance=1.5em,every node/.style={draw,circle,inner sep=2.5pt}]
    \node[right=of i,xshift=-1em] (n1) {};
    \node[right=of n1] (n2) {};
    \node[right=of n2,xshift=2em] (n3) {};
    \node[below=of n3] (n4) {};
    \node[right=of n4] (n5) {};
    \draw (n5) circle (2pt);
    \node[above=of n3,yshift=1em] (n6) {};
    \node[above left=of n6] (n7) {};
    \node[below left=of n6] (n8) {};    
    \node[left=of n6,xshift=-2em] (n9) {};
    \node[above left=of n9] (n10) {};
  \end{scope}
  \begin{scope}[->,>=latex']
    \draw (i) edge (n1);
    \draw (n1) edge (n2);
    \draw (n2) edge (n3);
    \draw (n2) edge[bend right] (n4);
    \draw (n4) edge node[above] {$E_2$} (n5);
    \draw (n3) edge[bend right] (n6);
    \draw (n6) edge[out=80,in=0] (n7);
    \draw (n6) edge[out=270,in=0] (n8);
    \draw (n7) edge[bend right=10] (n9);
    \draw (n8) edge[bend left=10] node[below] {$E_1$} (n9);
    \draw (n9) edge[out=60,in=350] (n10);
    \draw (n10) edge[out=270,in=160] (n9);
    \draw (n9) edge[bend right] (n2);
  \end{scope}
\end{tikzpicture}
\end{tabular}
\vskip-0.75em

In $\Thp(E)$, when processing an input string $w$, the run will first choose the prioritized choice of the union operator (which is $\eps$), iterate the inner loop once, and then return to the initial state of the sub-pNFA corresponding to $\eps|E_1$. Now, the first alternative is blocked, meaning that Algorithm~\ref{algo:match} tries to match $E_1$. Assuming that no failure occurs, it will then proceed by following $\epsilon$ transitions leading to $E_2$.

Now look at $\Jp(E)$. Here, the run first bypasses $E_1$, similarly to $\Thp(E)$, but this leads to the state following the start state. As the first alternative of transitions leaving this state has already been used, the run drops out of the loop and proceeds with $E_2$. $E_1$ will only be tried after backtracking in case $E_2$ fails.

We thus get several cases by appropriately instantiating $E_1$ and $E_2$.
Assume first that we choose $E_1$ in such a way that $\Thp(E_1)$ suffers from exponential backtracking on a set $W$ of input strings over $\Sigma$, and $E_2=\Sigma^*$. Then $\Thp(E)$ causes exponential backtracking on strings in $W$ whereas $\Jp(E)$ does not backtrack at all. A concrete example is obtained by taking $\Sigma=\{a,b\}$, $E_1=(a^*)^*$, and $W=\{a^nb\mid n\in\nat\}$.

Conversely, we may choose $E_2=\eps\Or E_2'$ so that $\Jp(E_2')$ fails exponentially on $W$, but $E_1=\Sigma^*$. Then $\Thp(E)$ will match strings in $W$ in linear time whereas $\Jp(E)$ will take exponential time.

One can easily combine two examples of the types above into one, to obtain an expression such that $\Thp(E)$ shows exponential behavior on a set $W$ of strings on which $\Jp(E)$ runs in linear time whereas $\Jp(E)$ shows exponential behavior on another set $W'$ of strings on which $\Thp(E)$ runs in linear time.

\section{Static Analysis of Exponential Backtracking}

We now consider the problem of deciding whether a given pNFA causes
backtracking matching similar to Algorithm~\ref{algo:java-match} to
run exponentially. More precisely, we ask whether a pNFA has
exponentially large backtracking runs. In the case where the
considered pNFA is $\Jp(E)$, this yields a statement about the running
time of Algorithm~\ref{algo:java-match}. However, we are interested in
the problem in general, because other regular expression engines may
correspond to other pNFA. There are two variants of the decision
problem, with very different complexities. Let us start by defining
the first.

\begin{definition}  
  \label{defn:constr-static}
  Given a pNFA $A = (Q_1,Q_2,\sig,q_0,\delta_1,\delta_2,F)$, let $f(n)
  = \max \{ |\btr_A(w)| \mid w \in \sig^*, |w| \le n\}$ for all $n\in
  \nat$. We say that $A$ has \emph{exponential backtracking} if $f\in
  2^{\Omega(n)}$ (or equivalently, if $f(n) \in 2^{\Theta(n)}$) and \emph{polynomial backtracking of degree $k$} for
  $k\in\nat$ if $f\in \Theta(n^{k+1})$.
  
  If the pNFA $A^f=(Q_1,Q_2,\sig,q_0,\delta_1,\delta_2,\emptyset)$,
  has exponential backtracking (or polynomial backtracking), then we say that
  $A$ has  \emph{exponential failure backtracking} (polynomial failure backtracking, resp.).
\end{definition}

Failure backtracking provides an upper bound for the general case.
In cases where the worst-case matching complexity 
can be exhibited by a family of strings not in $\L(A)$, this analysis is precise. This 
happens for example if for some $\$ \in\Sigma$, we have $w\$ \not\in\L(A)$ for all $w\in\Sigma^*$, or more generally, if for each $w\in\Sigma^*$, there is $w'\in\Sigma^*$ such that
$ww'\not\in\L(A)$. Failure backtracking analysis is of great interest in
  that it is more efficiently decidable (being in PTIME) than the
  general case.
  It is closely related to the case
  considered in e.g.~\cite{kirrage:13}, where the matching complexity of the strings not in $\L(A)$ is studied.

\subsection{An Upper Bound on the Complexity of General Backtracking Analysis}

Let us first establish an upper bound on the complexity of general
backtracking analysis. We will give an algorithm which solves this
problem in EXPTIME. Afterwards, we will also note some minor hardness
results. The EXPTIME decision procedure relies heavily on a result
from~\cite{drewes:01}.
\begin{lemma}
  \label{lem:exp-output-decide}
  Given a string-to-tree transducer $\stt=(Q,\sig,\Gamma,q_0,\delta)$, it is decidable in deterministic exponential time whether the function $f(n) = \max\{|t| \mid t\in\stt(s),\ s\in\sig^*,\ |s| \le n\}$ grows exponentially, i.e.\ whether $f\in 2^{\Omega(n)}$.
\end{lemma}

In short, we will hereafter construct a string-to-tree transducer from
a pNFA $A$ which reads an input string (suitably decorated) and
outputs the corresponding backtracking run of $A$
(see Definition~\ref{defn:btr}).  In this way, we model the running of
Algorithm~\ref{algo:match} on that string. Then
Lemma~\ref{lem:exp-output-decide} can be applied to this transducer to
decide exponential backtracking. To simplify the construction we first
make a small adjustment to the input pNFA in the form of a
``flattening'', which ensures that $\delta_2$ maps $Q_2$ to
$Q_1^*$. That is, we remove the opportunity for repeated
$\epsilon$-transitions.

\begin{definition}  
  \label{defn:delta2-flat}
  Let $A = (Q_1,Q_2,\sig,q_0,\delta_1,\delta_2,F)$ be a pNFA. Define
  $d\colon (Q_1\cup Q_2) \times (Q_2 \to \nat) \to Q_1^*$, and
  $\bar{r}\colon Q_1^* \to Q_1^*$ as follows:
  \[
  d(q, C) = \left\{ \begin{array}{cl}
      q & \text{if $q \in Q_1$,}\\
      d(q_{i+1},C_{q\mapsto i+1}) \cdots d(q_n,C_{q \mapsto i+1}) & \text{if $q \in
        Q_2$, $\delta_2(q) = (q_1 \cdots q_n)$ and $C(q) = i$.}
    \end{array}
  \right.
  \]
  \[
  \bar{r}(s) = \left\{\begin{array}{cl}
      \bar r(uv) & \text{if $s=uqv$ for some $u,v\in Q_1^*$ and $q\in Q_1$ with $|u|_q\ge2$}\\
      s & \text{otherwise.}
    \end{array}\right.
  \]
  That is, $\bar{r}$ removes all repetitions of each state $q$ beyond
  the first two occurrences.

  Now, the \emph{$\delta_2$-flattening of $A$} is the pNFA
  $A'=(Q_1,Q_2,\sig,q_0,\delta_1,\delta_2',F')$ with $\delta_2'(q) =
  \bar{r}(d(q,0^{Q_2}))$ for all $q\in Q_2$, and $F' = \{q \in Q_1
  \cup Q_2 \mid d(q,0^{Q_2})\cap F\neq\emptyset\}$.
\end{definition}

First let us note that the size of $A'$ in Definition~\ref{defn:delta2-flat} is polynomial in
  the size of $A$, as no new states are added and no right-hand side
  is greater than polynomial in length ($2|Q_1|$ is the maximum length
  after applying $\bar{r}$). Furthermore, the construction itself can be performed
  in polynomial time in a straightforward way by computing $d$
  incrementally in a left-to-right fashion, and aborting each
  recursion visiting a state that has already been seen twice to the
  left.

Before proving some properties of the above construction we make a
supporting observation.

\begin{lemma}  
  \label{lemma:slender-exp-tree}
  Let $\sigma$ be a function on trees such that, for $t=f[t_1,\dots,t_k]$
  \[
  \sigma(t)=\left\{\begin{array}{ll}
  t&\text{if $k=0$}\\
  f[\sigma(t_1)]&\text{if $k=1$}\\
  f[\sigma(t_i),\sigma(t_j)]&\text{otherwise, where $t_i,t_j$ ($i\neq j$) are largest among $t_1,\dots,t_k$}.
  \end{array}\right.
  \]
  Let $T_0,T_1,T_2,\dots$ be sets of trees of rank at most $k$. Then
  the function $f(n)=\max\{|t|\mid t\in T_n\}$ grows exponentially if
  and only if $f'(n)=\max\{|\sigma(t)|\mid t\in T_n\}$ grows
  exponentially.
\end{lemma}

We leave out the (rather easy) proof of the lemma due to space limitations.

\begin{lemma}
  \label{lemma:flat-is-ok}
  Let $A=(Q_1,Q_2,\sig,q_0,\delta_1,\delta_2,F)$ be a pNFA and $A'$ its $\delta_2$-flattening. Then $A'$ can be constructed in polynomial time, $\L(A')=\L(A)$, and
  the function $f(n)=\max\{|\btr_A(w)|\mid w\in\sig^*,\ |w|\le n\}$ grows
  exponentially if and only if $f'(n)=\max\{|\btr_{A'}(w)|\mid w\in\sig^*,\ |w|\le n\}$ grows exponentially.
\end{lemma}

\begin{proofsketch}
  Let $A'=(Q_1,Q_2,\sig,q_0,\delta_1,\delta_2',F')$. As noted, $A'$ can be constructed in polynomial time.
  
  The language equivalence of $A$ and $A'$ can be established by
  induction on the accepting runs of $A$ and $A'$. $\delta_2'$
  is a closure on $\delta_2$, such that any accepting run for $A$ of
  the form $p_1 \cdots p_n$ can be turned into one for $A'$ by
  replacing each maximal subsequence $p_{k} \cdots p_{k+i} \in Q_2^*$
  with just $p_k$. The function $d$ in the construction of $\delta_2$
  will ensure that $p_k$ is accepting if this was at the end of the
  run, and that $p_k$ can go directly to the following $Q_1$ state.
  The converse is equally straightforward, as a suitable sequence from $Q_2$
  can be inserted into an accepting run for $A'$ to create a correct
  accepting run for $A$.

  Finally, we argue that $A'$ exhibits exponential backtracking
  behavior if and only if $A$ does. By the construction of $A'$, we have $\btr_{A'}(w)\le\btr_A(w)$. Hence, $f$ grows exponentially if $f'$ does. It remains to consider the other direction. Thus, assume that $f(n)$ grows exponentially. We have to show that $f'(n)$ grows exponentially as well. Let $A''$ be the pNFA generated by $\delta_2$-flattening $A$ without applying $\bar{r}$. Let $t=\btr_A(w)$ and $t''=\btr_{A''}(w)$ for some input string $w$. Then $t''$ is obtained from $t$ by repeatedly replacing subtrees of the form $q[s_1,\dots,s_k,q'[t_1,\dots,t_l],s_{k+1},\dots,s_m]$, where $q,q'\in Q_2$, by $q[s_1,\dots,s_k,t_1,\dots,t_l,s_{k+1},\dots,s_m]$. Since Definition~\ref{defn:btr} prevents repeated $\eps$-cycles, this process removes only a constant fraction of the nodes in $t$.\footnote{The constant may be exponential in the size of $A$, but for the question at hand this does not matter since the backtracking behavior in the length of the string is what is considered.} Hence, $f''(n)=\max\{|\btr_{A''}(w)|\mid w\in\sig^*,\ |w|\le n\}$ grows exponentially. Now, compare $t''$ with $t'=\btr_{A'}(w)$. If a node of $t''$ has $m$ children with the same state $q\in Q_2$ in their roots, by the definition of backtracking runs the $m$ subtrees rooted at those nodes will be identical. This is the case since the run for each subtree starts in the same state and string position, and the application of $d$ in the partial flattening ensures that $C$ is made irrelevant by an immediately following $\delta_1$ transition resetting it to $0^{Q_2}$. The application of $\bar r$ to $A''$ means that, in effect, the first two copies of these $m$ subtrees are kept in $t'$. In particular, the two largest subtrees of the node are kept in $t'$. According to Lemma~\ref{lemma:slender-exp-tree}, this means that $g'$ grows exponentially.
  \end{proofsketch}
  
  It should be noticed that, for the proof above to be valid, it is important that $\bar{r}$ preserves the
  order of occurrences of states from the left, as a subtree being accepting means that no further subtrees are
  constructed to the right of it (ensuring no extraneous subtrees get included).
  

We are now prepared to define the
construction which for any $\delta_2$-flattened pNFA $A$ produces a
string-to-tree transducer $\stt$ such that $\btr_A(w) = t$ if and only if $t\in \stt(w')$. Here, $w'$ is a version of $w$ decorated with extra symbols $\sps$ and $\$$. The former will serve as padding to be read when $\delta_2$ transitions are taken, and $\$$ marks the beginning and the end of the
string.
\begin{definition}
  \label{defn:transduc}
  Given a $\delta_2$-flattened pNFA
  $A=(Q_1,Q_2,\sig,q_0,\delta_1,\delta_2,F)$ we construct the string-to-tree
  transducer $\stt = (Q,\sig',\Gamma,q_0',\delta)$ in the
  following way. $Q = \{q_0'\} \cup \{a_q,f_q \mid q\in Q_1 \cup Q_2\}$,  $\sig' = \sig \uplus \{\sps,\$\}$, and $\Gamma = Q_1 \cup Q_2 \uplus \{\acc,\rej\}$. Furthermore, $\delta$ consists of the following transitions:
    \begin{enumerate}
    \item Let $q_0' \xto{\$} a_{q_0}$ and $q_0' \xto{\$}
      f_{q_0}$. For all $q\in Q$ let $q\xto{\sps} q$.
    \item For all $q\in Q_1$ and $\alpha \in \Sigma$:
      \begin{enumerate}
      \item If $\delta_1(q,\alpha)=q'$ let $a_q \xto{\alpha}
        q[a_{q'}]$ and $f_q \xto{\alpha} q[f_{q'}]$.
      \item If $\delta_1(q,\alpha)$ is undefined let $f_q \xto{\alpha}
        q[\rej]$.
      \end{enumerate}
    \item For all $q \in Q_2$, if $q_1 \cdots q_n = \delta_2(q)$,
      then for all $i\in \{0,\ldots,n-1\}$ let $a_q \xto{\sps}
      q[f_{q_1},\ldots,f_{q_i},a_{q_{i+1}}]$, and let $f_q
      \xto{\sps} q[f_{q_1},\ldots,f_{q_n}]$.
    \item Finally if $q\in F$ let $a_q \xto{\$} q[\acc]$, whereas when
      $q \notin F$:
      \begin{enumerate}
      \item if $q\in Q_1$ let $f_q \xto{\$} q[\rej]$, and,
      \item if $q \in Q_2$ and $q_1 \cdots q_n = \delta_2(q)$, then
        $f_q \xto{\$} q[q_1[\rej],\ldots,q_n[\rej]]$.
      \end{enumerate}
  \end{enumerate}
\end{definition}

\begin{definition}
  The string $w_1 \alpha_1 w_2 \alpha_2 \cdots w_n \alpha_n w_{n+1}$ is a
  \emph{decoration} of $\alpha_1\cdots \alpha_n \in \Sigma^*$ if $w_i \in
  \{\$,\sps\}^*$ for each $i$. $\$\sps \alpha_1 \sps \alpha_2 \cdots \sps
  \alpha_n\$$ is the \emph{correct decoration} of $\alpha_1\cdots
  \alpha_n$, denoted $\textit{dec}(\alpha_1\cdots \alpha_n)$.
\end{definition}

\begin{lemma}
  \label{lemma:transduc-is-ok}
  For a $\delta_2$-flattened pNFA $A$, the string-to-tree transducer
  $\stt$ as constructed by Definition~\ref{defn:transduc}, and an
  input string $w=\alpha_1 \cdots \alpha_n$, it holds that
  $\stt(\textit{dec}(w)) = \{\btr_A(w)\}$. For \emph{all} $u$ which
  are decorations of $w$ either $\stt(u) = \emptyset$ or $\stt(u) =
  \{\btr_A(w)\}$.
\end{lemma}

\begin{proofsketch}
  First, notice how $A$ being $\delta_2$-flattened impacts
  $\btr_A$. The flattening ensures that there is no way to take two
  $\eps$-transitions in a row in $A$, meaning that every time case~2 of Definition~\ref{defn:btr}
  applies, we have $C(q) = 0$ since the previous
  step is either the initial call or a call from case~1 where $C$ gets
  reset. As such we will have $C = 0^{Q_2}$ in every recursive call
  below. Let $\stt_q$ denote the string-to-tree transducer $\stt$ with the initial
  state $q$ (instead of $q_0$).

  Let $v=\$\sps \alpha_1\sps \alpha_2 \sps \cdots \sps \alpha_n
  \$$. Establishing that $\stt(\textit{dec}(w)) = \{\btr_A(w)\}$ merely requires
  a straightforward case analysis the details of which we leave out due to space limitations.
  Starting with the case where the backtracking run on $w$
  \emph{fails}, the analysis establishes that for rejecting backtracking runs
  $t=\btr_A(q,w,0^{Q_2})$, we have $t\in \stt_{f_q}(v)$, for all $q$,
  where $v$ equals $\textit{dec}(w)$ with the initial $\$$ removed (we will
  deal with this at the end) and, vice versa, $t\in
  \stt_{f_q}(v)$ is true for exactly one $t$, so $t=\btr_A(q,w,0^{Q_2})$.

  The proof for the accepting runs follows very similar lines, but with
  the extra wrinkle of how $Q_2$ rules are handled when some path
  accepts. The invariant that $t \in \stt_{a_q}(v)$ is true for at
  most one $t$ is maintained however, as is, of course, the parallel
  to $\btr_A$. Again, the proof shows that $\stt_{a_q}(v)$ outputs precisely one tree if $v$ is $\textit{dec}(w)$
  with the initial $\$$ removed. That initial $\$$ is now used by the
  initial rules in $\stt$: $q_0' \xto{\$} a_{q_0}$ and $q_0' \xto{\$}
  f_{q_0}$. This means that $\stt$ produces exactly one tree for every
  $\textit{dec}(w)$, and in both the accepting and rejecting case it matches the
  tree from $\btr_A$.

  Finally, we need to deal with \emph{incorrect} decorations. Let $v$
  be a decoration of $w$ which is not $\textit{dec}(w)$. If $v$ has no
  leading $\$$, or no trailing $\$$, or has a $\$$ in any other
  position, $\stt(v)=\emptyset$, since $\stt$ has no other possible
  rules for $\$$. If $v$ contains \emph{extraneous} $\sps$ we still
  have $\stt(v) = \{\btr_A(w)\}$, since they will just be consumed by
  $q \xto{\sps} q$ rules. If some $\sps$ is ``missing'' compared to
  $\textit{dec}(w)$ this either causes $\stt(v) = \emptyset$, if a
  $Q_2$ rule needed it, or $\stt(v) = \{\btr_A(w)\}$, if it is just
  removed by a $q\xto{\sps} q$ rule anyway.
\end{proofsketch}

\begin{theorem}
  \label{thm:exptime-dec}
  It is decidable in exponential time whether a given pNFA $A$ has
  exponential backtracking.
\end{theorem}
\begin{proof}
  From $A$, construct the $\delta_2$-flattened pNFA $A'$ according to
  Definition~\ref{defn:delta2-flat}. According to
  Lemma~\ref{lemma:flat-is-ok}, $A'$ can be constructed in polynomial time,
   and it has exponential backtracking if
  and only if $A$ has. Construct the transducer $\stt$ for $A'$ according to
  Definition~\ref{defn:transduc}. By Lemma~\ref{lemma:transduc-is-ok}
  $\stt$ outputs exponentially large trees if and only if
  $A'$ has exponential backtracking. The construction of $\stt$ can clearly
  be implemented to run in polynomial time. Hence,
  Lemma~\ref{lem:exp-output-decide} yields the result.
\end{proof}

\subsection{Hardness of General Backtracking Analysis}

It seems likely that general backtracking analysis is computationally
difficult. We cannot prove this yet, but here
we demonstrate that either it is hard to decide if $\Jp(E)$ has exponential backtracking \emph{or} the class of
regular expressions $E$ such that $\Jp(E)$ does \emph{not} have exponential backtracking
has an easy universality decision problem. In the following, we say that $E$ has exponential backtracking if $\Jp(E)$ does.

Let us briefly recall the universality problem.

\begin{definition}  
  A regular expression $E$ is \emph{$\sig$-universal} if $\Sigma^*\subseteq\L(E)$.
  The input of \emph{RE Universality} is an alphabet $\Sigma$ and a regular expression $E$ over $\Sigma$. The question asked is whether $\L(E)$ is $\sig$-universal.
\end{definition}

This problem is well-known to be PSPACE-complete. See
e.g.~\cite{garey:79}. We will now give a simple polynomial reduction
which takes a regular expression $E$ and constructs a new regular
expression $E'$ such that $E'$ has exponential backtracking if $E$ has
exponential backtracking \emph{or} $E$ is
not universal.

\begin{lemma}  
  Let $E$ be a regular expression over $\sig$, $\alpha\in\sig$, and $\Gamma=\sig\cup\{\$\}$ for some $\$\notin\sig$.
  If $E$ does not have exponential backtracking then $E'= ((E\Or E\$\Gamma^*)\Or (\Sigma^*\$(\alpha^*)^*\$)$
  has exponential backtracking if and only if $E$ is not
  $\sig$-universal.
\end{lemma}

%
\begin{proof}
  If $E$ does not have exponential backtracking then
  neither does $E\$\Gamma^*$, since $\Gamma^*$ never fails. Now, let $A=\Jp(E')$.
  For every input string, the backtracking run of $A$ will attempt to match
  $\Sigma^*\$(\alpha^*)^*\$$ to the string only if neither $E$ nor
  $E\$\Gamma^*$ matches it. If $E$ is universal, i.e.\ equal to
  $\Sigma^*$, then $\L(E|(E\$\Gamma^*))=\L(\Sigma^*|(\Sigma^*\$\Gamma^*))=\Gamma^*$
  (since a string in $\Gamma^*$ is either in $\Sigma^*$ or has a prefix in $\sig^*$
  followed by a suffix in $\Gamma^*$ that begins with a $\$$). Hence, in this case
  $E'$ has exponential backtracking if and only if $E$ does.

  If we instead assume that $E$ is \emph{not} universal, then there
  exists some $w\in \Sigma^*$ such that $w \notin \L(E)$. Consider
  the string $w\$\alpha^n$ for any $n\in \nat$. Neither $E$
  nor $E\$\Gamma^*$ matches it, which means that backtracking will proceed into
  $\Sigma^*\$(\alpha^*)^*\$$, where $2^n$ backtracking attempts will
  be made to match the suffix $\alpha^n\$$ to the subexpression $(\alpha^*)^*\$$
   (as the final $\$$ keeps failing to match).
\end{proof}

The previous lemma yields the following corollary.

\begin{corollary}
  Let $\mathcal E$ be the set of all regular expressions that do \emph{not}
  have exponential backtracking. Then either RE Universality is not PSPACE-hard for inputs in $\mathcal E$, \emph{or} deciding whether regular expressions have exponential backtracking is PSPACE-hard.
\end{corollary}

\subsection{The Complexity of Failure Backtracking Analysis}

Now we look at the problem to decide whether a given pNFA has
exponential failure backtracking (see
Definition~\ref{defn:constr-static}).  For reasons of technical simplicity, assume that parallel $\eps$-transitions are absent from pNFA in this section.
To simplify the exposition in this section, and to obtain a useful notion of ambiguity for NFA with $\eps$-cycles, we restrict our notion of accepting runs of an NFA, 
as originally defined in Section~\ref{prelim}. Consider a run $p_1\cdots p_{m+1}$ on an input string $w=\beta_1\cdots\beta_m\in\sig^*$. This run is called \emph{short} if there are no $i,j$, $1\le i < j \le m$, such that $\beta_i=\ldots=\beta_{j}=\eps$, $p_i=p_j$, and $p_{i+1} = p_{j+1}$.  Thus, a short run must not contain any $\eps$-cycle in which an $\eps$-transition appears twice.

First we recall definitions from \cite{Mohri:08} on ambiguity for NFA, but for NFA with $\eps$-cycles. These definitions differ from those in  \cite{Mohri:08}, due to the fact that we allow $\eps$-cycles by using short accepting runs.
We define
the \emph{degree of ambiguity} of a string $w$ in $N$, denoted by $\da(N, w)$, 
to be the number of short accepting runs in $N$ labeled by $w$. 
$N$ is \emph{polynomially ambiguous} if there exists a polynomial $h$ such that $\da(N, w) \le h(|w|)$ for all $w\in\sig^*$. The minimal degree of such a polynomial is the 
\emph{degree of polynomial ambiguity} of $N$. We call $N$ \emph{exponentially ambiguous} if $g(n)=\max_{|w|\le n}\da(N, w) \in 2^{\Omega(n)}$ (or equivalently, if $g(n) \in 2^{\Theta(n)}$). It follows from Proposition~1 of \cite{Mohri:08} that $N$ is either polynomially or exponentially ambiguous, i.e., there is nothing in between. To be precise, this concerns only NFA without $\eps$-cycles, but as the proof of the following theorem shows, it extends to our more general case.

\begin{theorem}\label{thm:ambiguity}
 For an NFA $N$ it is decidable in time $O(|N|^3_E)$ whether $N$ is polynomially ambiguous, where $|N|_E$ denotes the number of transitions of $N$. If $N$ is polynomially ambiguous, the degree of
 polynomial ambiguity can be computed in time $O(|N|^3_E)$.
\end{theorem}

\begin{proof}
If $N$ is $\eps$-cycle free, the result follows from Theorems 5 and 6 in \cite{Mohri:08}. Now let $N=(Q,\sig,q_0,\delta, F)$ be an NFA, potentially with $\epsilon$-cycles, and define the equivalence relation $\sim$ on $Q$, where 
$p\sim q$ if and only if they are in the same strongly connected component determined by using only $\eps$-transitions in $N$.
Let $N' := N/\sim$ be the quotient of $N$ by $\sim$, having as states the equivalence classes of $\sim$.

The correctness of the remainder of the argument requires $N$ not to have equivalence classes with two elements, say $p, q$, where both $p$ and $q$ do not have $\eps$ self-loops. 
We briefly argue how equivalences classes of this form can be removed without changing the ambiguity properties of $N$. 
It is tedious, but straightforward, to verify that this can for example be achieved by replacing $p$ and $q$ (and $p\xrightarrow\eps q$,  $q\xrightarrow\eps p$) with 6 states and the appropriately defined $\eps$-transitions to model
the behavior of short runs in $N$ that go through one or both consecutively of $p$ and $q$. Three of the 6 states are used to model incoming transitions to $p$ in (short) runs that after reaching $p$ do not follow $p\xrightarrow\eps q$, or follow 
only $p\xrightarrow\eps q$, or follow consecutively $p\xrightarrow\eps q$ and $q\xrightarrow\eps p$, and the other 3 states are used for $q$ in a similar way. 

$N'$ could potentially have (parallel) $\eps$ self-loops. Let $N''$ be $N'$ with $\eps$ self-loops removed.
Each state in $N''$ will belong to exactly one of the following categories of equivalence classes:
(a) a single state of $N$ without an $\eps$ self-loop in $N$;
(b) a single state of $N$ with an $\eps$ self-loop in $N$;
(c) at least two states such that, in $N$, there are at least two distinct $\eps$-runs (staying within the equivalence class)  between any two states in the equivalence class (thanks to the modification of $N$ described in the preceding paragraph). 

Let $Z$ be the set of states in $N''$ having the properties specified in (b) or (c).
In $N''$ there are two possibilities.
Either (i) there is a (short run which is a) cycle in $N''$ having at least one state in $Z$, or (ii) each short run in $N''$ goes through at most $k$ states in $Z$ 
($k$ is bounded by the number of states in $N''$).
In case (i), $N$ is exponentially ambiguous, since we have at least two $\eps$-runs in $N$ between any two states in an equivalence class in $Z$.
In case (ii), the number of accepting runs in $N''$ (by definition without $\eps$-cycles) and number of short accepting runs in $N$, 
differ by a constant factor, and we can apply the $\eps$-cycle free result from \cite{Mohri:08} to $N''$.
\end{proof}

\begin{theorem}\label{thm:unconstraint}
A pNFA $A$ has either polynomial or exponential failure backtracking. It can be decided in time $O(|A|^3_E)$ whether $A$ has polynomial failure backtracking, and if so, the degree of backtracking can be computed in time $O(|A|^3_E)$.

\end{theorem}

\begin{proof}
Recall that $A^f$ denotes the pNFA obtained from $A$ where we change all states of $A$ so that they are not accepting, and
$\nfa A^f$ denotes the NFA obtained by ignoring priorities on transitions of $A^f$.
For an NFA $N$, $a(N)$ is obtained from $N$ by adding a new accepting sink state $z$ (having transitions to itself on all input letters), 
 all other states in $N$ are made non-accepting, and we add
$\eps$-transitions from all states in $N$ to $z$.
Since  $\da(a(\nfa A^f), w) = |\btr_{A^f}(w)|$, and thus
$ \max \{  \da(a(\nfa A^f), w) \mid w \in \sig^*, |w| \le n\} = \max \{ |\btr_{A^f}(w)| \mid w \in \sig^*, |w| \le n\}$,  the failure backtracking complexity of $A$ is 
equal to the ambiguity of $a(\nfa A^f)$.
To complete the proof, apply Theorem~\ref{thm:ambiguity} to $a(\nfa A^f)$.
\end{proof}





\section{Conclusion/Future Work}

Our prioritized NFA model is the only automata model, that we are aware of, which formalizes backtracking regular expression matching. This model is well suited to be extended to describe 
notions such as possessive quantifiers, captures and backreferences found in practical regular expressions.
Backreferences have been formalized in~\cite{Campeanu:03}, but without eliminating ambiguities due to multiple matches. 
Trying to improve our current complexity result for deciding backtracking complexity (as in Definition~\ref{defn:constr-static}), and 
secondly, to formalize what is meant by equivalence of a regular expression with a pNFA, will provide the impetus for future investigations.

\paragraph*{Acknowledgment} We thank the referees for extensive lists of valuable comments.
\bibliographystyle{eptcs}
\bibliography{main}

\comment{
\newpage
\appendix
\section{More Detailed Proofs}

In this appendix, we collect some proofs that had to be omitted or substantially shortened in the main part of the paper. For convenience, we re-state the results these proofs belong to.

\begin{faketheorem}{Lemma}{lemma:slender-exp-tree}
  Let $\sigma$ be a function on trees such that, for $t=f[t_1,\dots,t_k]$
  \[
  \sigma(t)=\left\{\begin{array}{ll}
  t&\text{if $k=0$}\\
  f[\sigma(t_1)]&\text{if $k=1$}\\
  f[\sigma(t_i),\sigma(t_j)]&\text{otherwise, where $t_i,t_j$ ($i\neq j$) are largest among $t_1,\dots,t_k$}.
  \end{array}\right.
  \]
  Let $T_0,T_1,T_2,\dots$ be sets of trees of rank at most $k$. Then
  the function $f(n)=\max\{|t|\mid t\in T_n\}$ grows exponentially if
  and only if $f'(n)=\max\{|\sigma(t)|\mid t\in T_n\}$ grows
  exponentially.
\end{faketheorem}

\begin{proof}
  The \emph{if} direction is obvious. Let us consider the \emph{only if} direction. Without loss of generality, we may assume that every node of a tree in $\bigcup_{n\in\nat}T_i$ is either a leaf or has at least two children (i.e.\ the second case of the definition of $\sigma(t)$ never applies). Now, the implication to be proved is equivalent to saying that $g'(n)=\max\{\ell(\sigma(t))|\mid t\in T_n\}$ grows exponentially if $g(n)=\max\{\ell(t)|\mid t\in T_n\}$ does, where $\ell(t)$ denotes the number of leaves of a tree $t$. However, among all trees $t$ of rank $k$ with a given number of leaves, the balanced $k$-ary trees $t$ are those which minimize $\ell(\sigma(t))$. Clearly, for such a tree, assuming for simplicity that it is fully balanced, we have $\ell(\sigma(t))=\ell(t)^c$, where $c=\log_k2$, which proves the statement.
\end{proof}

\begin{faketheorem}{Lemma}{lemma:transduc-is-ok}
  For a $\delta_2$-flattened pNFA $A$, the string-to-tree transducer
  $\stt$ as constructed by Definition~\ref{defn:transduc}, and an
  input string $w=\alpha_1 \cdots \alpha_n$, it holds that
  $\stt(\textit{dec}(w)) = \{\btr_A(w)\}$. For \emph{all} $u$ which
  are decorations of $w$ either $\stt(u) = \emptyset$ or $\stt(u) =
  \{\btr_A(w)\}$.
\end{faketheorem}

\begin{proof}
  First, notice how $A$ being $\delta_2$-flattened impacts
  $\btr_A$. The flattening ensures that there is no way to take two
  $\eps$-transitions in a row in $A$, meaning that every time case~2 of Definition~\ref{defn:btr}
  applies, we have $C(q) = 0$ since the previous
  step is either the initial call or a call from case~1 where $C$ gets
  reset. As such we will have $C = 0^{Q_2}$ in every recursive call
  below. Let $\stt_q$ denote the string-to-tree transducer $\stt$ with the initial
  state $q$ (instead of $q_0$).

  Let $v=\$\sps \alpha_1\sps \alpha_2 \sps \cdots \sps \alpha_n
  \$$. The proof will simply be a lengthy case analysis. For the most
  part the cases in the definition of $\btr_A$ have a very direct
  one-to-one relationship with what is done in $\stt$, with some
  details requiring clarification. We work our way backwards, starting
  with the subtrees generated when the empty string remains, $w'=\eps$
  and $v=\$$.

  We start with the case where the backtracking run on $w$
  \emph{fails}. We divide this into two cases.
  \begin{itemize}
  \item Let the remaining suffixes of $w$ and $v$ be $w'=\eps$ and
    $v'=\$$. For every state $q$ of $A$
    we have:
    \begin{itemize}
    \item When $q\in Q_1\setminus F$, let $t=q[\rej]$. Then we have
      $\btr_A(q,w,0^{Q_2}) = t$ and $t \in \stt_{f_q}(v')$.
    \item When $q\in Q_2 \setminus F$ and $q_1 \cdots q_n =
      \delta_2(q)$, let $t=q[q_1[\rej],\ldots,q_n[\rej]]$. Then it holds that
      $\btr_A(q,w',0^{Q_2}) = t$ and $t \in \stt_{f_q}(v')$.
    \end{itemize}
    These are by construction. Notice that in both cases $t$ is the
    \emph{only} tree in $\stt_{f_q}(v)$.
  \item Let the remaining suffixes of $w$ and $v$ be
    $w'=\alpha_1\cdots \alpha_n$ for $n>0$ and $v' \in \{
    \alpha_1 \sps \cdots \sps \alpha_n \$, \sps \alpha_1 \sps
    \cdots \sps \alpha_n\$\}$. For every state $q$ of $A$ we have:
    \begin{itemize}
    \item When $q\in Q_1$ and $\delta_1(q,\alpha_1) = q'$ then
      $\btr_A(q,w',0^{Q_2}) = q[\btr_A(q',\alpha_2 \cdots
      \alpha_n,0^{Q_2})]$ and, if $t'\in \stt_{f_{q'}}(\sps \alpha_2 \cdots \sps
      \alpha_n \$)$ then $q[t'] \in
      \stt_{f_q}(v')$. The rule used in the transducer is given by
      case~2(a) of Definition~\ref{defn:transduc}, preceded by $f_q
      \xto{\sps} f_q$ if $v'$ has a leading $\sps$.
    \item When $q \in Q_1$ and $\delta_1(q,\alpha_1)$ is undefined we
      have $\btr_A(q,w',0^{Q_2}) = q[\rej]$ and $q[\rej] \in
      \stt_{f_q}(v')$.  The rule used in the transducer is given by
      case~2(b) of Definition~\ref{defn:transduc}, preceded by $f_q
      \xto{\sps} f_q$ if $v'$ has a leading $\sps$.
    \item When $q\in Q_2$ and $\delta_2(q) = q_1\cdots q_n$, since we
      assume the backtracking run to fail all $q_i$ paths will fail
      and we will have $\btr_A(q,w',0^{Q_2}) =
      q[\btr_A(q_1,w',0^{Q_2}), \ldots, \btr_A(q_n,w', 0^{Q_2})]$. In
      the transducer a rule from case 3 is applied to get
      $q[t_1,\ldots,t_n] \in \stt_{f_q}(v')$ where $t_i \in
      \stt_{f_{q_i}}(\alpha_1 \sps \cdots \sps \alpha_n)$ for each
      $i$. Notice that here a leading $\sps$ in $v$ is
      \emph{required}, but one will always be available (as we cannot
      have two $Q_2$ states in a row due to
      $\delta_2$-flattening). Since none of the $q_i$ are in $Q_2$ the
      next step will read the leading $\alpha_1$ and a new $\sps$ will
      be available in the next recursive step.
    \end{itemize}
    Notice that each step in $\stt$ creates precisely one tree if only
    a single subtree is the recursive call generates only a single
    tree, and as shown above the base case produces only a single
    tree. As such $t \in \stt_{f_q}(v)$ is by induction true for
    exactly one tree for each (properly decorated) $v$. This property
    will be maintained for $a_q$ states as well.
  \end{itemize}
  This establishes that for rejecting backtracking runs
  $t=\btr_A(q,w,0^{Q_2})$, we have $t\in \stt_{f_q}(v)$, for all $q$,
  where $v=\textit{dec}(w)$ with the initial $\$$ removed (we will
  deal with this at the end) and, vice versa, $t\in
  \stt_{f_q}(v)$ is true for exactly one $t$, so $t$ must be the tree $\btr_A(q,w,0^{Q_2})$.

  The proof for the accepting runs follows very similar lines, but with
  the extra wrinkle of how $Q_2$ rules are handled when some path
  accepts. The invariant that $t \in \stt_{a_q}(v)$ is true for at
  most one $t$ is maintained however, as is, of course, the parallel
  to $\btr_A$. Take a $w$ on which the backtracking run of $A$
  \emph{succeeds}. We divide this into the same two cases.
  \begin{itemize}
  \item Let the remaining suffixes of $w$ and $v$ be $w'=\eps$ and
    $v'=\$$ respectively. For every $q\in F$ we have $t=q[\acc]$ and
    $\btr_A(q,w',0^{Q_2})=t$ and $t \in \stt_{a_q}(v')$. This
    completes all cases for the empty input string.
  \item Let the remaining suffixes of $w$ and $v$ be $w'=\alpha_1
    \cdots \alpha_n$ for $n>0$, and $v' \in \{\sps 
    \alpha_1 \cdots \sps \alpha_n\$, \alpha_1 \sps \cdots \sps
    \alpha_n\$$ respectively. For every state $q$ of $A$ we have:
    \begin{itemize}
    \item The case where $q\in Q_1$ and $\delta_1(q,\alpha_1) = q'$ is
      precisely like the failing case except with $a_q$ and $a_{q'}$
      in place of $f_q$ and $f_{q'}$.
    \item The case where $q \in Q_1$ but $\delta_1(q,\alpha_1)$ cannot
      give rise to a backtracking run that succeeds.
    \item When $q\in Q_2$ and $\delta_2(q) = q_1\cdots q_n$ we get
      something slightly more complex. One of the paths $q_i$ will
      accept, by assumption, so $\btr_A(q,w',0^{Q_2}) =
      q[\btr_A(q_1,w',0^{Q_2}), \ldots, \btr_A(q_i,w',0^{Q_2})]$ by
      construction. Case~3 of Definition~\ref{defn:transduc}
      makes it possible to generate any
      $q[f_{q_1},\ldots,f_{q_{i-1}},a_{q_i}]$, so here
      $q[t_1,\ldots,t_i] \in \stt_{a_q}(v')$ with $t_j \in
      \stt_{f_{q_j}}(\alpha_1 \sps \cdots \sps \alpha_n\$)$ for $j\in
      \{1,\ldots,i-1\}$ and $t_i \in \stt_{a_{q_i}}(\alpha_1 \sps
      \cdots\sps \alpha_n \$)$. This is in fact the \emph{only}
      possible case, as the transducer, as here sketched out, can only
      successfully complete its computation from a state $a_q$ if $q$
      eventually accepts the string, and can only complete its
      computation from a state $f_q$ state if $q$ fails.
    \end{itemize}
  \end{itemize}
  Again, this shows that $\stt_{a_q}(v)$ outputs precisely one tree if $v$ is $\textit{dec}(w)$
  with the initial $\$$ removed. That initial $\$$ is now used by the
  initial rules in $\stt$: $q_0' \xto{\$} a_{q_0}$ and $q_0' \xto{\$}
  f_{q_0}$. This means that $\stt$ produces exactly one tree for every
  $\textit{dec}(w)$, and in both the accepting and rejecting case it matches the
  tree from $\btr_A$.

  Finally, we need to deal with \emph{incorrect} decorations. Let $v$
  be a decoration of $w$ which is not $\textit{dec}(w)$. If $v$ has no
  leading $\$$, or no trailing $\$$, or has a $\$$ in any other
  position, $\stt(v)=\emptyset$, since $\stt$ has no other possible
  rules for $\$$. If $v$ contains \emph{extraneous} $\sps$ we still
  have $\stt(v) = \{\btr_A(w)\}$, since they will just be consumed by
  $q \xto{\sps} q$ rules. If some $\sps$ is ``missing'' compared to
  $\textit{dec}(w)$ this either causes $\stt(v) = \emptyset$, if a
  $Q_2$ rule needed it, or $\stt(v) = \{\btr_A(w)\}$, if it is just
  removed by a $q\xto{\sps} q$ rule anyway.
\end{proof}
}
\end{document}